\documentclass{paper}
\usepackage[margin=1in]{geometry}
\usepackage{amsmath}
\usepackage{amsthm}
\usepackage{amssymb}
\usepackage{framed}

\usepackage[utf8x,utf8]{inputenc} 
\usepackage[T1]{fontenc}
\usepackage{graphicx}
\usepackage{todonotes}
\usepackage{mathrsfs}
\usepackage{hyperref}

\usepackage[noend]{algorithmic}
\usepackage[boxed]{algorithm}

\newcommand{\algorithmiccommentt}[1]{\colorbox{black!10}{#1}}
\newcommand{\LineIf}[2]{ \STATE \algorithmicif\ {#1}\ \algorithmicthen\ {#2} }

\DeclareMathOperator{\E}{\mathbb{E}}
\newcommand{\subsetsum}{\textsc{Subset Sum}}
\newcommand{\knapsack}{\textsc{Knapsack}}

\newcommand{\B}{\mathcal{B}}
\newcommand{\Z}{\mathbb{Z}}

\newcommand{\Fs}{\mathscr{F}}

\newcommand{\cF}{\mathcal{F}}
\newcommand{\cL}{\mathcal{L}}
\newcommand{\cR}{\mathcal{R}}

\DeclareMathOperator{\poly}{poly}
\newcommand{\Os}{O^*\!}

\newcounter{openquestion}
\newenvironment{openquestion}
{\begin{center}\begin{minipage}{.95\textwidth}\begin{framed}
        \refstepcounter{openquestion}
        \textbf{Open Question~\theopenquestion:}
      }
{\end{framed}\end{minipage}\end{center}}
		
\newtheorem{theorem}{Theorem}[section]
\newtheorem{lemma}[theorem]{Lemma}
\newtheorem{definition}[theorem]{Definition}
\newtheorem{proposition}[theorem]{Proposition}
\newtheorem{claim}[theorem]{Claim}
\newtheorem{fact}[theorem]{Fact}

\newtheorem{corollary}[theorem]{Corollary}

\numberwithin{equation}{section}

\newtheorem*{rep@theorem}{\rep@title}
\newcommand{\newreptheorem}[2]{%
\newenvironment{rep#1}[1]{%
 \def\rep@title{#2 \ref{##1}}%
 \begin{rep@theorem}[restated]}%
 {\end{rep@theorem}}}

\makeatother

\newreptheorem{theorem}{Theorem}
\newreptheorem{lemma}{Lemma}
\newreptheorem{proposition}{Proposition}

\begin{document}

\author{Per Austrin \and Petteri Kaski \and Mikko Koivisto \and Jesper Nederlof}
\title{
Dense Subset Sum may be the hardest\\[2mm]{}
}

\maketitle
\setcounter{page}{0}
\thispagestyle{empty}

\begin{abstract}
The \subsetsum{} problem asks whether a given set of $n$ positive integers contains a subset of elements that sum up to a given target $t$.
It is an outstanding open question whether the $\Os(2^{n/2})$-time 
algorithm for \subsetsum{} by Horowitz and Sahni [J.~ACM~1974] can be beaten in the worst-case setting by a ``truly faster'', $\Os(2^{(0.5-\delta)n})$-time algorithm, with some constant $\delta > 0$.
Continuing an earlier work [STACS 2015], we study \subsetsum{} parameterized by the maximum bin size $\beta$, defined as the largest number of subsets of the $n$ input integers that yield the same sum.  For every $\epsilon > 0$ we give a truly faster algorithm for instances with $\beta \leq 2^{(0.5-\epsilon)n}$, as well as instances with $\beta \geq 2^{0.661n}$.
Consequently, we also obtain a characterization in terms of the popular density parameter $n/\log_2 t$: if all instances of density at least $1.003$ admit a truly faster algorithm, then so does every instance. This goes against the current intuition that instances of density 1 are the hardest, and therefore is a step toward answering the open question in the affirmative.
Our results stem from novel combinations of earlier algorithms for \subsetsum{} and a study of an extremal question in additive combinatorics connected to the problem of Uniquely Decodable Code Pairs in information theory.
\end{abstract}

\noindent
{\em Author affiliations and addresses:}

\bigskip

\noindent
Per Austrin\\
School of Computer Science and Communication\\
KTH Royal Institute of Technology, Sweden\\
\texttt{austrin@csc.kth.se}\\

\medskip
\noindent
Petteri Kaski\\
Helsinki Institute for Information Technology HIIT 
\& Department of Computer Science\\
Aalto University, Finland\\
\texttt{petteri.kaski@aalto.fi}\\

\medskip
\noindent
Mikko Koivisto\\
Helsinki Institute for Information Technology HIIT \& Department of Computer Science\\
University of Helsinki, Finland\\
\texttt{mikko.koivisto@helsinki.fi}\\

\medskip
\noindent
Jesper Nederlof\\
Department of Mathematics and Computer Science\\
Technical University of Eindhoven, The Netherlands\\
\texttt{j.nederlof@tue.nl}\\

\clearpage

\section{Introduction}

The \subsetsum{} problem and its generalization to the \knapsack{} problem are two of the most famous NP-complete problems. In the \subsetsum{} problem, we are given positive integers $w_1,w_2,\ldots,w_n,t \in \Z$ as input, and need to decide whether there exists a subset $X\subseteq[n]$ with $\sum_{j\in X} w_j=t$. In the \knapsack{} problem, we are additionally given integers $v_1,v_2,\ldots,v_n$ and are asked to find a subset $X\subseteq [n]$ maximizing $\sum_{j \in X}v_j$ subject to the constraint $\sum_{j \in X}w_j \leq t$. While the study of \subsetsum{} is, among others, motivated by cryptographic applications or balancing problems, \knapsack{} has numerous applications in combinatorial optimization.
%
%
We study the exact worst-case time complexity of these problems. The earliest and probably most important algorithms for both problems are simple applications of dynamic programming, pioneered by Bellman~\cite{Bellman57}, solving both problems in $\Os(t)$ time (where the $\Os(\cdot)$ notation suppresses factors polynomial in the input size). In terms of $n$, the best algorithms for both problems are due to Schroeppel and Shamir~\cite{SchroeppelShamir81}, using $\Os(2^{n/2})$ time and $\Os(2^{n/4})$ space, based on the \emph{meet-in-the-middle} technique by Horowitz and Sahni~\cite{HorowitzSahni74}. Nederlof et al.~\cite{DBLP:conf/mfcs/NederlofLZ12} show that there is an $\Os(T^n)$-time, $\Os(S^n)$-space algorithm for \subsetsum{} if and only if there is an $\Os(T^n)$-time, $\Os(S^n)$-space algorithm for \knapsack{}. A major open question since the paper by Horowitz and Sahni~\cite{HorowitzSahni74} is whether we can do ``truly faster'' for both problems:

\begin{openquestion}
  \label{q:fast}
  Can \subsetsum{} be solved in $\Os\big(2^{(0.5 - \delta) n}\big)$ time for some constant $\delta > 0$?
\end{openquestion}

In this paper we discuss Monte Carlo
algorithms in the following sense: the algorithm never returns false positives and constructs solutions of yes-instances with at least inverse polynomial probability. All randomized algorithms discussed in this paper are of this type, 
but for Open Question~\ref{q:fast} we would be satisfied
with two-sided error as well.

Zooming out, one motivation of this question is as follows. It is commonly believed that there are no polynomial time or even sub-exponential time algorithms for \subsetsum{}. So how fast can the fastest algorithm be? It would be an elegant situation if the simple meet-in-the-middle algorithm was optimal. But this would also be quite surprising, and  so we aim to show that at least this is not the case.

In 2010, Howgrave-Graham and Joux~\cite{HowgraveGrahamJoux10} gave an algorithm that answered Open Question 1 in the affirmative in an \emph{average case} setting. To state their result, let us describe the setting where it applies.  The \emph{density} of a $\subsetsum{}$ instance is defined as $n / \log_2 t$. A random instance of density $d>0$ is constructed by fixing $t\approx 2^{n/d}$ and picking the integers $w_1, \ldots, w_n, t$ independently and uniformly at random between $1$ and $2^{n/d}$. Howgrave-Graham and Joux~\cite{HowgraveGrahamJoux10} showed that random instances of density $1$ can be solved in $\Os(2^{0.311n})$ time, and later this has been improved to $\Os(2^{0.291n})$ time by Becker et al.~\cite{BeckerCoronJoux11}. These results resolve Open Question 1 in the average case setting since Impagliazzo and Naor~\cite{impagliazzonaor} showed that random instances are the \emph{hardest when they have density $1$}. Indeed, a vast body of research has given better algorithms for random instances with density deviating from $1$, like reductions of sparse instances to the shortest vector problem (e.g.~\cite{DBLP:journals/jacm/LagariasO85,Coster92improvedlow-density}) and the algorithm by Flaxman and Przydatek~\cite{DBLP:conf/stacs/FlaxmanP05}.

The algorithms discussed thus far all use exponential space, which can be a serious bottleneck. Therefore many studies also emphasize the setting where the algorithm is restricted to using polynomial space. %
It is known that the running time of the dynamic programming based
algorithms can be achieved also in polynomial space: Lokshtanov and
Nederlof~\cite{LokshtanovN10} give polynomial space algorithms solving
\subsetsum{} in $\Os(t)$ time and \knapsack{} in
pseudo-polynomial time.  On the other hand, in terms of $n$, no polynomial space algorithm significantly faster than na\"\i{}vely going through all $2^n$ subsets is known, and the following has been stated as an open problem by a number of researchers~(see e.g.~\cite{Woeginger2008,openprob}):
\begin{openquestion}
\label{q:polyspace}
Can \subsetsum{} be solved in polynomial space and $\Os\big(2^{(1-\delta) n}\big)$ time for some constant $\delta > 0$?
\end{openquestion}

\subsection{Our results}

We aim to make progress on Open Question~\ref{q:fast}, and show that a large class of instances can be solved truly faster. An optimist may interpret this as an indication that truly faster algorithms indeed exist, while a pessimist may conclude the remaining instances must be the (strictly) hardest instances.

\paragraph{Algorithmic Results}

To define classes of instances that admit truly faster algorithms, we consider several natural parameters. The key parameter that seems to capture the range of our algorithmic technique the best is the \emph{maximum bin size} $\beta(w) = \max_{x\in\mathbb{Z}}|\{S \subseteq [n]: \sum_{i \in S} w_i = x\}|$. Our main technical result is:

\begin{theorem}\label{thm:smallbin}
  There exists a Monte Carlo algorithm that, for
  any $0\le \epsilon \le 1/6$, solves all instances of \subsetsum{}
  with $\beta(w)\leq 2^{(0.5-\epsilon)n}$ in
  $\Os\big(2^{(0.5-\epsilon/4+3\epsilon^2/4)n}\big)$ time.
\end{theorem}

We have not optimized the precise constants in
Theorem~\ref{thm:smallbin} -- the main message is that any instance
with bin size up to $2^{(0.5-\epsilon)n}$ can be solved in time
$2^{(0.5-\Omega(\epsilon))n}$.
For $\epsilon \ge 1/6$, the running time of $2^{23n/48}$
obtained for $\epsilon = 1/6$ is still valid since $2^{(0.5-1/6)n}$
remains an upper bound on $\beta(w)$. In a
previous work~\cite{DBLP:conf/stacs/AustrinKKN15}, we solved
\subsetsum{} in time $\Os(2^{0.3399n}\beta(w)^4)$, which is faster than
Theorem~\ref{thm:smallbin} for small $\beta(w)$, but
Theorem~\ref{thm:smallbin} shows that we can beat the meet-in-the-middle bound 
for a much wider class of instances. 

From the other end, we also prove that when the maximum bin size
becomes too large, we can again solve \subsetsum{} truly faster:

\begin{theorem}
  \label{thm:large bin easy}
  There exist a constant $\delta > 0$ and a deterministic algorithm that solves all instances of
  \subsetsum{} with $\beta(w) \ge 2^{0.661n}$ in $\Os\big(2^{(0.5 - \delta) n}\big)$ time.
\end{theorem}

\paragraph{Combinatorial Results} Given Theorem~\ref{thm:smallbin}, a natural question is how instances
with $\beta(w)\geq 2^{0.5n}$ look like. This question is an
instantiation of the inverse Littlewood-Offord problem, a subject
well-studied in the field of additive combinatorics. Ideally we would
like to find structural properties of instances with $\beta(w) \ge
2^{0.5n}$, that can be algorithmically exploited by other means than
Theorem~\ref{thm:smallbin} in order to resolve Open
Question~\ref{q:fast} in the affirmative.
While there is a
large amount of literature on the inverse Littlewood-Offord problem, the
typical range of $\beta(w)$ studied there is $\beta(w) =
2^n/\poly(n)$ which is not relevant for our purposes.
However, we did manage to
determine additional properties that any instance that is not solved
by Theorem~\ref{thm:smallbin} must satisfy.

In particular, we study a different natural parameter, the \emph{number of distinct sums} generated by $w$, defined as
$|w(2^{[n]})|=\{w(X) : X \subseteq [n] \}$ (where we denote
$w(X)=\sum_{i \in X}w_i$).  This parameter can be viewed as a measure
of the ``true'' density of an instance, in the following sense.  An
instance with density $d = n/\log_2 t$ has $|w(2^{[n]})| \le n
2^{n/d}$ (assuming without loss of generality that $t \le \max_i
w_i$).  On the other hand, by standard hashing arguments (e.g., Lemma~\ref{lemma:reduce t} with $B = 10|w(2^{[n]})|$), any instance
can be hashed down to an equivalent instance of density roughly
$n/\log_2 |w(2^{[n]})|$.

The relationship between $|w(2^{[n]})|$ and $\beta(w)$ is more
complicated.
Intuitively, one would expect that if one has so much
concentration that $\beta(w) \ge 2^{0.5n}$, then $w$ should not
generate too many sums.
We are not aware of any such results from
the additive combinatorics literature.  However, by establishing a new
connection to \emph{Uniquely Decodable Code Pairs}, a well-studied object
in information theory, we can derive the following bound.

\begin{lemma}\label{lem:sumsvsbin}
  If $|w(2^{[n]})| \ge 2^{0.997n}$ then $\beta(w) \le 2^{0.4996n}$.
\end{lemma}

Unfortunately, we currently do not know how to algorithmically exploit $|w(2^{[n]})| \le 2^{0.997n}$. But we do know how to exploit a set $S$ with $|S|\leq n/2$ and $|w(2^S)|\leq 2^{0.4999n}$ (see Lemma~\ref{lemma:exploiting few sums}). This suggests the question of how large $\beta(w)$ can be in instances lacking such an $S$, and we prove the following bound.

\begin{lemma}\label{lemma:large bin ub}
  There is a universal constant $\delta > 0$ such that the following
  holds for all sufficiently large $n$.  Let $S, T$ be a partition of
  $[n]$ with $|S| = |T| = n/2$ such that $|w(2^S)|, |w(2^T)| \ge
  2^{(1/2-\delta)n}$.  Then $\beta(w) \le 2^{0.661n}$.
\end{lemma}

\paragraph{Further Consequences}
Combining Lemma~\ref{lem:sumsvsbin} and Theorem~\ref{thm:smallbin}, we
see directly that instances that generate almost $2^n$ distinct sums can be
solved faster than $2^{0.5n}$.

\begin{theorem}\label{cor:manysums}
  There exists a Monte Carlo algorithm that solves all instances of \subsetsum{} with $|w(2^{[n]})| \geq 2^{0.997n}$ in time $O^*(2^{0.49991n})$.
\end{theorem}

Combining this with the view described above of $|w(2^{[n]})|$ as a
refined version of the density of an instance, we have the following
result, to support the title of our paper:

\begin{theorem}\label{thm:densityreduction}
  Suppose there exist a constant $\epsilon >0$ and an algorithm that solves all \subsetsum{} instances of density at least $1.003$ in time $\Os(2^{(0.5-\epsilon)n})$. Then there exists a Monte Carlo algorithm that solves \subsetsum{} in time $\Os\big(2^{\max\{0.49991, 0.5-\epsilon\}n}\big)$.
\end{theorem}

After the result by Howgrave-Graham and Joux~\cite{HowgraveGrahamJoux10}, this may be a next step towards resolving Open Question~\ref{q:fast}.  Intuitively, one should be able to exploit the fact that the integers in a dense instance have fewer than $n$ bits. For example, even if only the target is picked uniformly at random, in expectation there will be an exponential number of solutions, which can easily be exploited.\footnote{\label{footnote:manysoln}For example, assuming there are at least $2^{\sigma n}$ solutions for a constant $\sigma\geq 0$, use a dynamic programming table data structure to randomly sample the subsets in the congruence class $t\bmod q$ for $q$ a random prime with about $(1-\sigma)n/2$ bits within linear time per sample. A solution is found within $\Os(2^{(1-\sigma)n/2})$ samples with high probability.}.

Finally, let us note a somewhat curious consequence of our results.
As mentioned earlier, in the context of Open
Question~\ref{q:polyspace}, it is known that the $\Os(2^{n/d})$
running time for instances of density $d$ achieved through dynamic
programming can be achieved in polynomial space~\cite{LokshtanovN10} (see also 
\cite[Theorem 1(a)]{DBLP:conf/iwpec/KaskiKN12}).
Combining this with Corollary~\ref{cor:manysums} and hashing, we directly get the following
``interleaving'' of Open Questions~\ref{q:fast} and \ref{q:polyspace}.

\begin{corollary}\label{cor:strange}
There exist two Monte Carlo algorithms, 
one running in $\Os(2^{0.49991n})$ time and 
the other in $\Os(2^{0.999n})$ time and polynomial space, 
such that every instance of \subsetsum{} is solved by at least one of the algorithms. 
\end{corollary}

\paragraph{Organization of the paper}
This paper is organized as follows: In Section~\ref{sec:prel} we
review some preliminaries. In
Section~\ref{sec:algorithms}, we provide the proofs of our main algorithmic results. 
In Section~\ref{sec:combinatorics} we prove two combinatorial lemmas.
In Section~\ref{sec:density} we give the proof for Theorem~\ref{thm:densityreduction}.
Finally we end with some discussion on in
Section~\ref{sec:discussion}.

\section{Preliminaries}\label{sec:prel}

For a modulus $m\in\Z_{\geq 1}$ and $x,y\in\Z$, we write $x\equiv y\pmod m$, or $x\equiv_m y$ for short, to indicate that $m$ divides $x-y$. Throughout this paper, $w_1, w_2,\ldots,w_n, t$ will denote the input integers of a \subsetsum{} instance. We associate the set function $w:2^{[n]}\rightarrow \mathbb{Z}$ with these integers by letting $w(X)=\sum_{i\in X}w_i$, and for a set family $\cF \subseteq 2^{[n]}$ we write $w(\cF)$ for the image $\{w(X) : X\in \cF\}$. 



For $0 \leq x_1,x_2,\ldots, x_\ell \leq 1$ with $\sum_{i=1}^{\ell}x_i =1$ we write $h(x_1,x_2,\ldots,x_\ell)=\sum_{i=1}^\ell -x_i\log_2 x_i$ for the entropy function. Here, $0\log_2 0$ should be interpreted as $0$. We shorthand $h(x,1-x)$ with $h(x)$. 
We routinely use the standard fact (easily proved using Stirling's formula) that for non-negative integers $n_1, \ldots, n_\ell$ (where $\ell$ is a constant) summing to $n$, it holds that $\binom{n}{n_1, \ldots, n_\ell} = 2^{h(n_1/n, \ldots, n_\ell/n) n} \cdot \poly(n)$.






\begin{claim}
\label{prop:random mod good}
For every sufficiently large integer $r$ the following holds. 
If $p$ is a prime between $r$ and $2r$ selected uniformly at random 
and $x$ is a nonzero integer, then $p$ divides $x$ with probability at most $(\log_2 x)/r$.  
\end{claim}
\begin{proof}
  By the Prime Number Theorem~\cite[p.~494, Eq.~(22.19.3)]{HardyWright08}, there are at least $r/\log_2 r$ primes
  between $r$ and $2r$.  A nonzero number $x$ can have at most $\log_r x$
  prime factors larger than $r$.  The probability that a random prime
  between $r$ and $2r$ is a factor of $x$ is therefore at most
  $(\log_r x)\big/(r/\log_2 r) = (\log_2 x)/r$.
\end{proof}

\begin{lemma}[Bit-length reduction]
  \label{lemma:reduce t}

There exists a randomized algorithm 
that takes as input a \subsetsum{} instance $w_1,w_2,\ldots,w_n,t\in\Z$ and an integer $B \in \Z$, 
and in time $O^*(1)$ outputs a new \subsetsum{} instance $w_1',w_2',\ldots,w_n', t'\in\mathbb{Z}$ such that with probability $\Omega^*(1)$, the following properties all simultaneously hold.
\begin{enumerate}
\item \label{hashlemma:modulus}
$0\leq w_1',w_2',\ldots,w_{n}',t'< 4n B \log_2 B$.
\item \label{hashlemma:correctness}
  If $B \ge 10 \cdot |w(2^{[n]})|$, then $X \subseteq [n]$ satisfies $w(X)=t$ 
  if and only if $w'(X)=t'$.
\item \label{hashlemma:sums}
  If $B \ge 10 \cdot |w(2^{[n]})|$, then $|w(2^{[n]})|/2 \leq|w'(2^{[n]})|\leq n |w(2^{[n]})|$.
\item \label{hashlemma:binsize}
  If $B \ge 5 \cdot |w(2^{[n]})|^2$, then $\beta(w)/n \leq \beta(w') \le \beta(w)$.
\end{enumerate}
\end{lemma}

\begin{proof}
  The algorithm picks a uniformly random prime $p$ from the interval
  $B \log_2 t \le p \le 2B \log_2 t$.  It then outputs $w_i' = w_i \bmod
  p$ for each $i = 1,2,\ldots,n$ and $t' = (t \bmod p) + r \cdot
  p$ for a uniformly randomly chosen $r \in \{0, \ldots, n-1\}$.  The
  number of bits in $p$ is $\Os(1)$, so the
  reduction runs in $\Os(1)$ time.

  This construction might not satisfy Property~\ref{hashlemma:modulus}
  if $t$ is very large -- instead it satisfies the weaker bound $t' <
  2n B \log_2 t$.  The desired bound can be obtained by repeating the
  reduction $O(1)$ times -- e.g., if after three steps $t''$ is more
  than $2n B \log_2 B$, then the original number $t$ was triply
  exponential and the input weights are so large that brute force time
  $\Os(2^n)$ is $\Os(1)$.
  
Consider any two sums $y_1,y_2\in w(2^{[n]})\cup\{t\}$ with $y_1\neq
y_2$.  We have $y_1\equiv_p y_2$ if and only if $p$ divides
$x=y_2-y_1\neq 0$. Since $0<|y_2-y_1|=|x|\leq 2 n t \le t^2$ (without
loss of generality we assume $t \ge 2n$ since otherwise dynamic
programming lets us solve the instance in polynomial time), by
Claim~\ref{prop:random mod good} we have $y_1\equiv_p y_2$ with
probability at most $(\log_2 |x|)/(B\log_2 t) \leq 2/B$.

To establish Property~\ref{hashlemma:correctness}, we take a union
bound over the $|w(2^{[n]})| \le B/10$ different sums
$y \in w(2^{[n]}) \setminus \{t\}$ and deduce that with
probability at least $4/5$ over the choice of $p$, $w(X) \ne t$ implies
$w'(X) \ne t'$.  In the other direction, if $w(X) = t$, then $w'(X) =
t'$ with probability exactly $1/n$ over the choice of $r$.

For Property~\ref{hashlemma:sums} the argument is similar: for any
given sum $y \in w(2^{[n]})$, the probability that there exists
another sum $y' \ne y$ that collides with $y$ mod $p$ is at most
$1/5$.  By Markov's inequality, this implies that with probability at
most $2/5$, more than half of all $y \in w(2^{[n]})$ collide with some
other sum $y'$ mod $p$.  Conversely with probability at least $3/5$, at least half of all $y \in w(2^{[n]})$ have no such collision, and in this case $|w'(2^{[n]})| \ge |w(2^{[n]})|/2$.  For the
upper bound, two sets $S_1$, $S_2$ with $w(S_1) = w(S_2)$ have
$w'(S_1) \equiv_p w'(S_2)$.  Given $w'(S_1) \bmod p$, there are at
most $n$ possible different values for $w'(S_1)$, so each bin of
$w(2^{[s]})$ is split into at most $n$ bins after the hashing.

Finally, for Property~\ref{hashlemma:binsize}, we use a union bound
over all pairs of sums in $\binom{w(2^{[n]})}{2}$.  Since $B \ge 5
|w(2^{[n]})|^2 \ge 10\binom{|w(2^{[n]})|}{2}$, with
probability at least $4/5$, no two sums $y_1, y_2 \in w(2^{[n]})$ are
hashed to the same value mod $p$ and thus $\beta(w') \le \beta(w)$.
For the lower bound, we have as before that each original bin is split
into at most $n$ bins after hashing, and the largest of these must be
at least a $1/n$ fraction of the original bin size.

Now we combine the above bounds to bound the probability all properties hold simultaneously: 
taking a union bound over the events depending on $p$, we have that
with probability at least $1/5$, $w(X) \ne t$ implies $w'(X) \ne
t'$, and Properties~\ref{hashlemma:sums} and \ref{hashlemma:binsize}
are satisfied.  Conditioned on this, the final good event, that $w(X)
= t$ implies $w'(X) = t$ still happens with probability $1/n$ since
$r$ is independent of $p$.
\end{proof}

\section{Algorithmic Results}
\label{sec:algorithms}

This section establishes Theorems~\ref{thm:smallbin} and~\ref{thm:large bin easy}. 
We begin with two lemmas showing how one can exploit a subset of the input integers if it generates either many or few distinct sums. The case of many sums is the main technical challenge and addressed by the following result, which is our main algorithmic contribution.  


\begin{lemma}\label{lem:binalgo}
  There is a randomized algorithm that, given positive integers
  $w_1,\ldots,w_n,t\leq 2^{O(n)}$ and a set $M\in \binom{[n]}{\mu n}$
  satisfying $\mu \leq 0.5$ and $|w(2^{M})|\geq 2^{\gamma|M|}$ for
  some $\gamma \in [0,1]$, finds a subset $X \subseteq [n]$ satisfying
  $w(X)=t$ with probability $\Omega^*\!(1)$ (if such an $X$ exists) in
  time $\Os\big(2^{\left(0.5+0.8113\mu-\gamma\mu\right)n} +
  \beta(w)2^{(1.5-\gamma)\mu n}\big)$.
\end{lemma}

The proof is given in Section~\ref{sect:binalgo proof}. Informally, it uses an algorithm that simultaneously applies the meet-in-the-middle technique of Horowitz and Sahni \cite{HorowitzSahni74} on the set $[n] \setminus M$ and the ``representation technique'' of Howgrave-Graham and Joux \cite{HowgraveGrahamJoux10} on the set $M$. Specifically, we pick an arbitrary equi-sized partition $L, R$ of $[n]\setminus M$ and construct lists $\cL \subseteq 2^{L \cup M}$ and $\cR \subseteq 2^{R \cup M}$. Note that without restrictions on $\cL$ and $\cR$, one solution $X$ is witnessed by $2^{|M \cap X|}$ pairs $(S,T)$ from $\cL \times \cR$ in the sense that $S \cup T = X$. Now the crux is that since $M$ generates many sums, $M\cap X$ generates many sums (say $2^{\pi |M|}$): this allows us to uniformly choose a congruence class $t_L$ of $\mathbb{Z}_p$ where $p$ is a random prime of order $2^{\pi|M|}$ and restrict attention only to sets $S \subseteq L \cup M$ and $T \subseteq R \cup M$ such that $w(S)\equiv_p t_L$ and $w(T)\equiv_p t-t_L$, while still finding solutions with good probability. This ensures that the to-be-constructed lists $\cL$ and $\cR$ are small enough. As an indication for this, note that if $|M \cap X|=|M|/2$ and $|w(\binom{M\cap X}{|M|/4})|$ is $\Omega(2^{|M|/2})$, the expected sizes of $\cL$ and $\cR$ are at most $2^{((1-\mu)/2+h(1/4)\mu-\mu/2)n}\leq 2^{(1/2-0.18\mu)n}$.

In contrast to Lemma~\ref{lem:binalgo}, it is straightforward to exploit a small subset that generates few sums:

\begin{lemma}
  \label{lemma:exploiting few sums}
  There is a deterministic algorithm that, given positive integers
  $w_1,\ldots,w_n,t$
  and a set $M \in \binom{[n]}{\mu n}$
  satisfying $\mu \leq 0.5$ and $|w(2^{M})|\leq 2^{\gamma|M|}$ for some $\gamma \in
  [0,1]$, finds a subset $X \subseteq [n]$ satisfying $w(X) = t$ (if such an $X$ exists) in time 
  $\Os\big(2^{\frac{1-\mu(1-\gamma)}{2} n}\big)$.
\end{lemma}  

\begin{proof}
  Let $L$ be an arbitrary subset of $[n] \setminus M$ of size
  $\frac{1-\mu(1-\gamma)}{2}n$ and let $R = [n]\setminus L$.
  Then $|w(2^L)| \le 2^{|L|} =2^{\frac{1-\mu(1-\gamma)}{2}n}$, and
  \[
  |w(2^R)| \le |w(2^M)|\cdot |w(2^{[n]\setminus L \setminus M})| \le 2^{\gamma \mu n} 2^{\left(1-\frac{1-\mu(1-\gamma)}{2} - \mu\right)n} = 2^{\frac{1-\mu(1-\gamma)}{2}n}.
  \]
  Now apply routine dynamic programming to construct $w(2^L)$ in time $O^*(|w(2^L)|)$ and $w(2^R)$ in time $O^*(|w(2^R)|)$; build a look-up table data structure for $w(2^L)$, and for each $x\! \in w(2^R)$, check in $O(n)$ time whether $t-x \in w(2^L)$.
\end{proof}
Given these lemmas, we are now in the position to exploit small bins:
\begin{proof}[Proof of Theorem~\ref{thm:smallbin}]
  We start by preprocessing the input with Lemma~\ref{lemma:reduce t}, taking $B = 2^{3n} \gg |w(2^{[n]})|^2$. 
  Let $\gamma=1-\epsilon/2, \mu=3\epsilon/2$, and partition $[n]$ into
  $1/\mu$ parts $M_1,\ldots,M_{1/\mu}$ of size at most $\mu n$ arbitrarily. We distinguish two cases.  First, suppose that $|w(2^{M_i})| \ge
  2^{\gamma \mu n}$ for some $M_i$ (note that this can be easily determined within the claimed time bound).  We then apply the algorithm of
  Lemma~\ref{lem:binalgo} with $M = M_i$ and solve the instance (with probability $\Omega^*(1)$) in
  time
  \[
  \Os\Big(2^{\left(0.5+0.8113\mu-\gamma\mu\right)n} + \beta(w) 2^{(1.5-\gamma)\mu n}\Big)\,.
  \]
  The coefficient of the exponent of the first term is 
$0.5 + 0.8113 \cdot 3\epsilon/2 - (1-\epsilon/2)\cdot 3\epsilon/2 
= 0.5 - 0.28305\epsilon + 0.75\epsilon^2$. 
The coefficient of the exponent of the second term is 
$0.5-\epsilon + (1.5-(1-\epsilon/2))\cdot 3\epsilon/2 
= 0.5 - \epsilon/4 + 0.75\epsilon^2$.

  Second, suppose that $|w(2^{M_i})| \le 2^{\gamma \mu n}$ for all
  $i$.  Let $L = \bigcup_{i=1}^{\frac{1}{2\mu}} M_i$ and $R = [n]
  \setminus M$.  We see that 
	%
	\[
    |w(2^L)| \leq \prod_{i \le \frac{1}{2\mu}} |w(2^{M_i})| \le 2^{\gamma n/2} 
    \qquad \textrm{and} \qquad   
    |w(2^R)| \leq \prod_{i  >  \frac{1}{2\mu}} |w(2^{M_i})| \le 2^{\gamma n/2}\,.
  \]
  Using standard dynamic programming to construct $w(2^L)$ and $w(2^R)$ in $O^*(|w(2^L)|)$ and $O^*(|w(2^R)|)$ time, we can therefore solve the instance within $\Os(2^{\gamma n /2}) = \Os(2^{(0.5-\epsilon/4)n})$ time using linear search.
\end{proof}
Exploiting large bins is easy using Lemma~\ref{lemma:large bin ub}:
\begin{proof}[Proof of Theorem \ref{thm:large bin easy}]
  Pick an arbitrary equi-sized partition $S, T$ of $[n]$.  By the
  contrapositive of Lemma~\ref{lemma:large bin ub}, one of $S$ and $T$
  generates at most $2^{(1/2-\delta)n}$ sums.  Applying
  Lemma~\ref{lemma:exploiting few sums} with the set in question as
  $M$, we get a running time of $\Os\big(2^{(1-\delta)n/2}\big)$.
\end{proof}
The proof of Lemma~\ref{lemma:large bin ub} is given in Section~\ref{sec:combinatorics}. 

\subsection{Proof of Lemma~\ref{lem:binalgo}}\label{sect:binalgo proof}

We now prove Lemma~\ref{lem:binalgo}.
Let $s := |X \cap M|$.  Without loss of generality, we may assume
that $s \ge |M|/2$ (by considering the actual target $t$ and the
complementary target $t' := w([n])-t$).  We may further assume that
$s$ is known by trying all $O(n)$ possible values.
The algorithm is listed in Algorithm~\ref{alg:manysums}.

\begin{algorithm}
  \caption{Exploiting a small subset generating many sums.}
  \label{alg:manysums}
  \begin{algorithmic}[1]
    \REQUIRE $\mathsf{A}(w_1,\ldots,w_n,t,M,s,\gamma)$\hfill\algorithmiccommentt{Assumes $|w(2^{M})|\geq 2^{\gamma|M|}$}
    \ENSURE $\mathbf{yes}$, if there exists an $X \subseteq [n]$ with $w(X)=t$ and $|X \cap M|=s$
    \STATE Let $\sigma = s/|M|$
    \STATE Let $\pi = \gamma-1+\sigma$ 
    \STATE Pick a random prime $p$ satisfying $2^{\pi |M|} \leq p \leq 2^{\pi |M|+1}$
    \STATE Pick a random number $0\leq  t_L \leq p-1$
    \FORALL{$0 \leq s_1 \leq s_2\leq |M|$ such that $s_1+s_2=s$}
    \STATE Let $\sigma_1 = s_1/|M|,\sigma_2 = s_2/|M|$ 
    \STATE Let $\lambda = (1-\mu)/2+\big(h(\sigma/2)-h(\sigma_1)\big)\mu$
    \STATE Let $L, R$ be an \emph{arbitrary} partition of $[n] \setminus M$ such that $|L|=\left\lceil \lambda n \right\rceil$
    \STATE Construct $\cL=\{ S \in 2^{L \cup M}: w(S) \equiv_p t_L \textrm{ and } |S \cap M| =s_1\}$ \label{alg1:build cL}
    \STATE Construct $\cR=\{ T \in 2^{R \cup M}: w(T) \equiv_p t-t_L \textrm{ and } |T \cap M|=s_2\}$ \label{alg1:build cR}
    \FORALL{$(S,T) \in \cL \times \cR$ such that $w(S) + w(T) = t$} \label{alg1:combine lists}
    \LineIf{$S \cap T =\emptyset$}{\algorithmicreturn\ $\mathbf{yes}$}
    \ENDFOR
    \ENDFOR
    \STATE \algorithmicreturn\ $\mathbf{no}$
  \end{algorithmic}
\end{algorithm}

\paragraph{Running time}

We focus on the expected running time of
Algorithm~\ref{alg:manysums} (by returning $\mathbf{no}$
after running the algorithm at least twice its expected number of steps, this is sufficient).
We will analyze the algorithm in two parts: (i) the
generation of the lists $\cL$ and $\cR$ on Lines~\ref{alg1:build cL}
and \ref{alg1:build cR}, and (ii) the iteration over pairs in $\cL
\times \cR$ in Line~\ref{alg1:combine lists} (the typical bottleneck).
Let $W_L := 2^{|L|}\binom{M}{s_1} \le 2^{\lambda n}
2^{h(\sigma_1)\mu n} = 2^{((1-\mu)/2 + h(\sigma/2)\mu)n}$ denote
the size of the search space for $\cL$.

\begin{proposition}
  The lists $\cL$ and $\cR$ in Lines~\ref{alg1:build cL} and
  \ref{alg1:build cR} can be constructed in expected time $\Os\big(W_L^{1/2} + W_L / 2^{\pi \mu n}\big)$, where the expectation is over the choice of $p$ and $t_L$.
\end{proposition}
\begin{proof}
  By splitting the search space for $\cL$ appropriately, we get two
  ``halves'' each of which has size $W_L^{1/2}$. Specifically, we arbitrarily pick a subset $L_1 \subseteq L$ of size $\lambda_1 n$ with $\lambda_1=(\lambda+h(\sigma/2)\mu)/2$ and generate using brute-force $w(2^{L_1})$ and $w(\cL_2)$ where $\cL_2=\{ Y \cup Z: Y \subseteq L \setminus L_1 \textrm{ and } Z \in \binom{M}{s_1}\}$. Then we store $w(2^{L_1})$ in a dictionary data structure and, for each sum $x \in w(\cL_2)$, we look up all solutions with sum
  $t-x \bmod p$ in the dictionary of $w(2^{L_1})$ and list for such a pair its union. This yields a
  running time of $\Os(|\cL| + W_L^{1/2})$.  The expected size of
  $|\cL|$ over the random choices of $t_L$ is $\E[|\cL|] \le  O(W_L/2^{\pi \mu n})$.

  The analysis for $\cR$ is analogous and we get a running time of $\Os\big(W_R^{1/2} +
  W_R/2^{\pi \mu n}\big)$ where $W_R := 2^{|R|} \binom{|M|}{s_2}$.  Let $\rho = |R|/n$.  Since $h(\cdot)$ is concave and, in particular, $h(\sigma_1) + h(\sigma_2) \le 2h(\sigma/2)$, we then have (up to a negligible term caused by rounding $\lambda n$ to an integer)
  \begin{align*}
    \rho &= 1-\mu-\lambda = (1-\mu)/2 - \big(h(\sigma/2)-h(\sigma_1)\big)\mu \le (1-\mu)/2 + (h(\sigma/2)-h(\sigma_2))\mu\,.
  \end{align*}
  Thus the case of $R$ is symmetric to the situation for $L$ and we have
  $W_R \le 2^{((1-\mu)/2 + h(\sigma/2)\mu)n} = O^*(W_L)$.
\end{proof}

The term $W_L/2^{\pi \mu n}$ can be bounded by using the definition of $\pi = \gamma -1 + \sigma$ and we get
\begin{align*}
  W_L/2^{\pi \mu n} = 2^{(\frac{1}{2} + \mu(\frac{1}{2} + h(\sigma/2) - \gamma - \sigma))n}\,.
\end{align*}
Since $1/2+h(\sigma/2)-\sigma$ subject to $1/2 \leq \sigma \leq 1$
is maximized at $\sigma=1/2$ where it is $h(1/4) \le 0.8113$, we
have that $W_L/2^{\pi \mu n} \le 2^{(0.5 + 0.8113\mu-\gamma\mu)n}$.

The term $W_L^{1/2}$ is naively bounded by $2^{(1+\mu)n/4}$, which is dominated by the term $\Os\big(2^{\left(0.5+0.8113\mu-\gamma\mu\right)n}\big)$ since $\mu \leq 1/2$ and $\gamma \le 1$. It follows that Line~\ref{alg1:build cL} and Line~\ref{alg1:build cR} indeed run within the claimed time bounds.

\begin{proposition}
  The expected number of pairs
  considered in Line~\ref{alg1:combine lists} is $\Os\big(\beta(w)
  2^{\mu(1.5-\gamma)n}\big)$, 
  where the expectation is over the choice of $p$ and $t_L$. 
\end{proposition}

\begin{proof}
  Define
  \[
  \B = \big\{ (P,Q) \in 2^{[n]} \times 2^M : w(P) + w(Q) = t\big\}\,,
  \]
  and note that the set of pairs $(S,T) \in 2^{L \cup M} \times 2^{R
    \cup M}$ satisfying $w(S) + w(T) = t$ are in one-to-one
  correspondence with pairs in $\B$ (by the map $(S, T) \mapsto (S
  \cup (T \cap R), T \cap M)$).  Furthermore, the size of $\B$ is
  bounded by $|\B| \le \beta(w) 2^{|M|}$: for each of the $2^{|M|}$
  possible choices of $Q$, there are at most $\beta(w)$ subsets $R$
  that sum to $t-w(Q)$.
  
  Any given pair $(S,T) \in 2^{L \cup M} \times 2^{R \cup M}$
  satisfying $w(L) + w(R) = t$ is considered only if $w(S) \equiv_p
  t_L$, which happens with probability $O(2^{-\pi n})$ (over the uniformly random choice of $t_L$).
  Thus the expected number of pairs considered in Line~\ref{alg1:combine lists} is upper bounded by
  $O(|B|/2^{\pi \mu n}) = O(\beta(w) 2^{\mu(1-\pi)n}) = O(\beta(w) 2^{\mu(2-\sigma-\gamma)n})$.
  Using $\sigma \ge 1/2$, the desired bound follows.
\end{proof}

\paragraph{Correctness}

Suppose there exists an $X \subseteq [n]$ with $w(X)=t$ and $|X \cap
M| = s$. Note that $2^{\gamma |M|}\leq |w(2^{M})| \leq |w(2^{M \cap
  X})| \cdot |w(2^{M \setminus X})|$, and since $|w(2^{M \setminus X})|\leq
2^{|M|-s}$, we have that $|w(2^{M \cap X})| \geq 2^{\gamma |M| -
  (1-\sigma)|M|}=2^{\pi |M|}$.

Thus there must exist positive $s_1+s_2=s$ such that
$|w(\binom{M\cap X}{s_1})| \geq 2^{\pi |M|}/|M|$. Let us
focus on the corresponding iteration of
Algorithm~\ref{alg:manysums}.  Let $w_L := w(X \cap L)$ be the
contribution of $L$ to the solution $X$.
We claim that in this iteration, the following holds.
\begin{proposition}
  \begin{equation}\label{eq:sumsmodprime}
    \Pr\left[\exists Q \in \binom{M \cap X}{s_1}: w(Q) \equiv_p t_L-w_L\right] \geq \Omega\bigg(\frac{1}{|M|}\bigg)\,.
  \end{equation}
\end{proposition}
Note that this is sufficient for establishing correctness of the
algorithm, since conditioned on this event,
Algorithm~\ref{alg:manysums} will include $S := Q \cup (L \cap X)$
in $\cL$ and $T := X \setminus S$ in $\cR$, and the algorithm recovers $X$.

\begin{proof}
  Let $\cF \subseteq \binom{M \cap X}{s_1}$
  be a maximal injective subset, i.e., satisfying $|\cF| = |w(\cF)| =
  |w\left(\binom{M \cap X}{s_1}\right)| \ge \Omega^*(2^{\pi |M|})$.  
  Let $c_i = |\left\{ Y \in \cF: w(Y) \equiv_p
  i\right\}|$ be the number of sets from $\cF$ in the $i$'th bin mod $p$.  Our goal is to lower bound the probability that $c_{t_L-w_L} > 0$ (where $t_L-w_L$ is taken modulo $p$).
  We can bound the expected $\ell^2$ norm (e.g., the number of collisions) by
  \begin{equation}
    \label{eqn:alg1 collision prob}
    \E\bigg[\sum_{i} c^2_i\bigg] = \sum_{Y,Z \in \cF}\Pr\big[p \text{ divides } w(Y)-w(Z)\big] \leq |\cF| + \Os\big(|\cF|^2 / 2^{\pi |M|}\big)\,,
  \end{equation}
  where the inequality uses Claim~\ref{prop:random mod good} and the
  assumption that the $w_i$'s are $2^{O(n)}$.
  By Markov's inequality, $\sum_{i} c^2_i \le \Os(|\cF|^2/2^{\pi
    |M|})$ with probability at least $\Omega^*(1)$ over the choice of
  $p$ (here we used $|\Fs| = \Omega^*(2^{\pi |M|})$ to conclude that
  the second term in \eqref{eqn:alg1 collision prob} dominates the
  first).  Conditioned on this, Cauchy-Schwarz implies that the number
  of non-zero $c_i$'s is at least $|\cF|^2\big/\sum_{i} c_i^2 \ge
  \Omega^*(2^{\pi |M|})$.  When this happens, the probability that
  $c_{t_L-w_L} > 0$ (over the uniformly random choice of $t_L$) is
  $\Omega^*(1)$.
\end{proof}

This concludes the proof of Lemma~\ref{lem:binalgo}.

\section{Combinatorial Results (Lemma~\ref{lem:sumsvsbin} and Lemma~\ref{lemma:large bin ub})}
\label{sec:combinatorics}

In this section we provide two non-trivial quantitative relations
between several structural parameters of the weights. Our results are
by no means tight, but will be sufficient for proving our main
results.

For the purposes of this section, it is convenient to use vector
notation for subset sums.  In particular, for a vector $x \in \Z^n$,
we write $x \cdot w = \sum_{i=1}^n x_i w_i$, and $x^{-1}(j) \subseteq
[n]$ for the set of $i \in [n]$ such that $v_i = j$.

Our approach to relate the number of sums $|w(2^{[n]})|$ to the
largest bin size $\beta(w)$ is to establish a connection to the notion
of Uniquely Decodable Code Pairs from information theory, defined as
follows.

\begin{definition}[Uniquely Decodable Code Pair, UDCP] If $A,B \subseteq \{0,1\}^n$ such that 
\[
	|A+B| = |\{a+b: a\in A, b \in B\}| = |A|\cdot|B|\,,
\]
then $(A,B)$ is called \textit{uniquely decodable}. Note that here addition is performed over $\mathbb{Z}^n$ (and \textbf{not} mod $\mathbb{Z}_2^n$).
\end{definition}

UDCP's capture the zero error region of the so-called \emph{binary
  adder channel}, and there is a fair amount of work on how large the
sets $A$ and $B$ can be (for a survey, see \cite[\S3.5.1]{schleger}).
The connection between UDCP's and \subsetsum{} is that a \subsetsum{}
instance that both generates many sums and has a large bin yields a
large UDCP, as captured in the following proposition.

\begin{proposition}
  \label{prop:udcp connection}
  If there exist weights $w_1, \ldots, w_n$ such that $|w(2^{[n]})| =
  a$ and $\beta(w) = b$, then there exists a UDCP $(A,B)$ with $|A| =
  a$ and $|B| = b$.
\end{proposition}

\begin{proof}
  Let $A \subseteq \{0,1\}^n$ be an injective set, i.e., $x\cdot w
  \neq x'\cdot w$ for all $x,x'\in A$ with $x\neq x'$. Note that there
  exists such an $A$ with $|A|=a$.  Let $B \subseteq \{0,1\}^n$ be a
  bin, i.e., $y\cdot w = y'\cdot w$ for all $y,y'\in B$.  Note that we
  can take these to have sizes $|A| = a$ and $|B| = b$.
  
  We claim that $(A,B)$ is a UDCP.  To see this, let $x,x'\in A$ and $y,y'\in B$ with $x+y=x'+y'$. Then 
  \[
  x \cdot w + y \cdot w = (x+y)\cdot w = (x'+y')\cdot w =x'\cdot w+ y' \cdot w\,.
  \]
  Thus $x \cdot w = x' \cdot w$, and so by the injectivity property of $A$,
  we have $x = x'$, which in turn implies $y=y'$ since $x+y=x'+y'$.
\end{proof}

We have the following result by Ordentlich and Shayevitz  \cite[Theorem~1, setting $R_1 = 0.997$ and $\alpha=0.07$]{DBLP:journals/corr/OrdentlichS14a}.
\begin{theorem}[\cite{DBLP:journals/corr/OrdentlichS14a}]
\label{thm:udcpbound}
Let $A,B \subseteq \{0,1\}^n$ such that $(A, B)$ is a UDCP and $|A|\geq 2^{.997n}$. Then $|B|\leq 2^{0.4996n}$.
\end{theorem}

With this connection in place, the proof of Lemma~\ref{lem:sumsvsbin}
is immediate.

\begin{replemma}{lem:sumsvsbin}
  If $|w(2^{[n]})| \ge 2^{0.997n}$, then $\beta(w) \le 2^{0.4996n}$.
\end{replemma}

\begin{proof}
  Combine Theorem~\ref{thm:udcpbound} with the contrapositive of
  Proposition~\ref{prop:udcp connection}.
\end{proof}

The remainder of this section is devoted to Lemma~\ref{lemma:large bin ub}:

\begin{replemma}{lemma:large bin ub}
  There is a universal constant $\delta > 0$ such that the following
  holds for all sufficiently large $n$.  Let $S, T$ be a partition of $[n]$ with $|S| = |T| = n/2$
  such that $|w(2^S)|, |w(2^T)| \ge 2^{(1/2-\delta)n}$.  Then $\beta(w)
  \le 2^{0.661n}$.
\end{replemma}

The proof also (implictly) uses a connection to Uniquely Decodable Code Pairs, but here the involved sets of strings are not binary. 
There is no reason to believe that the constant $0.661$ is tight. However, because
a random instance $w$ of density $2$ satisfies the hypothesis for all 
partitions $S,T$ and has $\beta(w)\approx 2^{0.5n}$ with good
probability, just improving the constant $0.661$ will not suffice for
settling Open Question~\ref{q:fast}. 

\subsection{Proof of Lemma~\ref{lemma:large bin ub}}

For a subset $S \subseteq [n]$, define a function $b_S: \Z
\rightarrow \Z$ by letting $b_S(x)$ be the number of subsets of $S$
that sum to $x$. 
Note that $|w(2^S)|$ equals the support size of
$b_S$, or $\|b_S\|_0$, and that $\beta_w(S) = \max_{x} b_S(x) =
\|b_S\|_{\infty}$.  Instead of working with these extremes, it is more convenient to work with the $\ell^2$ 
norm of $b_S$, and
the main technical claim to obtain Lemma~\ref{lemma:large bin ub} is
the following.
\begin{proposition}
  \label{prop:large bin l2}
  There exists a $\delta > 0$ such that for all sufficiently large $|S|$ the following holds: if $|w(2^{S})| \ge 2^{(1-\delta)|S|}$, then $\|b_S\|_2 \le 2^{0.661|S|}$.
\end{proposition}

Using Proposition~\ref{prop:large bin l2}, the desired bound of
Lemma~\ref{lemma:large bin ub} follows immediately, since
\[
\beta([n]) \,=\, \max_{x \in \Z} \sum_{y \in \Z} b_S(y)b_T(x-y) 
\,\le\, \max_{x \in \Z} \|b_S\|_2 \|b_T\|_2 
\,\le\, 2^{0.661n}\,,
\]
where the first inequality is by Cauchy--Schwarz and the second
inequality by Proposition~\ref{prop:large bin l2}.


\begin{proof}[Proof of Proposition~\ref{prop:large bin l2}]
  Without loss of generality we take $S = [n]$, and to simplify
  notation we omit the subscript $S$ from $b_S$ and simply write
  $b: \Z \rightarrow \Z$ for the function such that $b(r)$ is the number of subsets of $w_1, \ldots, w_n$ summing to $r$. Note that
  \begin{align*}
    \|b\|_2^2 & = \sum_{U,V \subseteq [n]} [w(U)=w(V)]
    = \sum_{\substack{U, V \subseteq [n]\\U \cap V = \emptyset}} [w(U) = w(V)] \cdot 2^{n-|U|-|V|}
    = \sum_{y \in \{-1,0,1\}^n} [y \cdot w = 0] \cdot 2^{|y^{-1}(0)|},
  \end{align*}
	where $[p]$ denotes $1$ if $p$ holds and $0$ otherwise. 
  Defining $B_\sigma = \big\{y \in \{-1,0,1\}^n \,:\, y \cdot w = 0 \textrm{ and } \|y\|_1 = \sigma n\big\}$, we thus have
  \begin{equation}
    \label{eqn:collision bound Bsigma}
  \|b\|_2^2 = \sum_{i=0}^n |B_{i/n}| 2^{n-i} \le n \max_{\sigma} |B_\sigma| 2^{(1-\sigma)n}\,.
  \end{equation}
  We now proceed to bound the size of $B_\sigma$ by an encoding
  argument. To this end, let $A \subseteq \{0,1\}^n$ be a maximal
  injective set of vectors.  In other words, $|A| = |w(2^{[n]})| \ge
  2^{0.99n}$, and for all pairs $x \ne x' \in A$, it holds that $x
  \cdot w \ne x' \cdot w$.  We claim that $|A + B_\sigma| = |A| \cdot
  |B_\sigma|$. To see this note that, similarly to the proof of Proposition~\ref{prop:udcp connection}, if $x+y = x'+y'$ (with
  $x,x'\in A$ and $y,y' \in B_\sigma$) then $x \cdot w = x'
  \cdot w$ (since $y' \cdot w = 0$) and thus $x = x'$ and $y =
  y'$.

  Define $P_{\sigma}$ to be all pairs $(x,y)$ in $A \times
  B_\sigma$ that are \emph{balanced}, in the sense that for some $\gamma >0$ the following conditions hold:
  \begin{equation}\label{eq:balancedcondition}
	\begin{aligned}
    |x^{-1}(1) \cap y^{-1}(-1)| &= \tfrac{1}{2}|y^{-1}(-1)| \pm \gamma n\,, \\
    |x^{-1}(1) \cap y^{-1}(0)| &= \tfrac{1}{2}|y^{-1}(0)| \pm \gamma n\,, \\
    |x^{-1}(1) \cap y^{-1}(1)| &= \tfrac{1}{2}|y^{-1}(1)| \pm \gamma n\,.
  \end{aligned}
  \end{equation}

  \begin{claim}\label{clm:balancedpairs}
	For $\gamma=\sqrt{\delta}$ and $n$ sufficiently large, we have that $|P_{\sigma}| \ge |A| \cdot |B_\sigma| /2 $.
  \end{claim}
  \begin{proof}
We prove the claim by giving an upper bound on the number of pairs that are not balanced. Note that if a pair $(x,y)$ is not balanced, there must exist a $j\in\{-1,0,1\}$ such that~\eqref{eq:balancedcondition} fails. Let us proceed with fixing a $y\in B_\sigma$ and upper bounding the number of $x \in A$ such that
\begin{equation}
  \label{eqn:unbalanced pair}
	\left| |x^{-1}(1) \cap y^{-1}(j)| - \tfrac{1}{2}|y^{-1}(j)| \right| > \gamma n\,.
  \end{equation}
Recall the following basic property of the binary entropy function.
  \begin{fact}
    \label{fact:binary entropy estimate around half}
    For all $\alpha \in \big[0,\tfrac{1}{2}\big]$, it holds that $h\big(\tfrac{1}{2} - \alpha\big) \le 1 - 2\alpha^2/\ln 2$.
  \end{fact}
Using this bound, the number of $x$ satisfying \eqref{eqn:unbalanced pair} is at most 
\[
	n\max_{\alpha \leq 1/2-\gamma }2^{h(\alpha)n} \leq n2^{(1 - 2\gamma^2/\ln 2)n}\,,
\]
and thus the number of pairs not balanced it at most $3n2^{(1 - 2\gamma^2/\ln 2)n}|B_\sigma|$. This is clearly at most a constant fraction of $|A|\cdot |B_\sigma|$ by the assumption $\gamma = \sqrt{\delta}$.	
	\end{proof}
		
  Setting $\gamma=\sqrt{\delta}$,
  we can now proceed to upper bound
  $|P_{\sigma}|$.  Consider the encoding $\eta: P_\sigma \rightarrow
  \{-1,0,1,2\}^n$ defined by $\eta(x,y) = x + y$.  By the
  property $|A + B_{\sigma}| = |A| \cdot |B|$, it follows that $\eta$
  is an injection, and thus $|P_{\sigma}|$ equals the size of the
  image of $\eta$.  For a pair $(x,y) \in P_\sigma$, if $y \in
  B_{\sigma}$ has $\tau \sigma n$ many $1$'s, and $(1-\tau)\sigma n$
  many $-1$'s, then $z = \eta(x,y)$ has the following frequency distribution:
  \begin{align*}
    \frac{|z^{-1}(-1)|}{n} &= \frac{\tau \sigma}{2} \pm o_\gamma(1)\,, &
    \frac{|z^{-1}(0)|}{n} &= \frac{\tau \sigma}{2} + \frac{1-\sigma}{2} \pm o_\gamma(1)\,, \\
    \frac{|z^{-1}(1)|}{n} &= \frac{1-\sigma}{2} + \frac{(1-\tau)\sigma}{2} \pm o_\gamma(1)\,, &
    \frac{|z^{-1}(2)|}{n} &= \frac{(1-\tau)\sigma}{2} \pm o_\gamma(1)\,,
  \end{align*}
	where, for a variable $\epsilon$, we write $o_\epsilon(1)$ to indicate a term that converges to $0$ when $\epsilon$ tends to $0$.
  Since $\gamma = \sqrt{\delta}$, we have $o_\gamma(1) = o_\delta(1)$.  The number of $z$'s with such a frequency distribution is bounded by
  \begin{equation}
    \label{eqn:z freq bound 1}
    \binom{n}{
      \tfrac{\tau \sigma}{2} n,
      (\tfrac{\tau\sigma}{2} + \tfrac{1-\sigma}{2})n,
      (\tfrac{1-\sigma}{2} + \tfrac{(1-\tau)\sigma}{2})n,
      \tfrac{(1-\tau)\sigma}{2}}
    2^{o_\delta(1)n}\,.
  \end{equation}
	
  Then, $|P_\sigma|$ is bounded by
  \[
  \log |P_{\sigma}| \le \max_{\tau \in [0,1]} \big(g(\sigma, \tau) + o_{\gamma}(1)\big)n, \text{ where } g(\sigma, \tau) = 
  h\Big(
    \tfrac{\tau \sigma}{2},
    \tfrac{\tau \sigma}{2} + \tfrac{1-\sigma}{2},
    \tfrac{1-\sigma}{2} + \tfrac{(1-\tau)\sigma}{2},
    \tfrac{(1-\tau)\sigma}{2}\Big)\,.
  \]
  It can be verified that $g(\sigma, \tau)$ is
  maximized for $\tau=1/2$ and we have
  \[
  \max_{\tau \in [0,1]} g(\sigma, \tau) = 
  h\big(
    \tfrac{\sigma}{4},
    \tfrac{1}{2}-\tfrac{\sigma}{4},
    \tfrac{1}{2}-\tfrac{\sigma}{4},
    \tfrac{\sigma}{4}\big) = 1 + h\big(\tfrac{\sigma}{2}\big)\,.
  \]
  Combining this with the bounds $|P_\sigma| \ge |A| \cdot |B| \cdot 2^{-O(\delta^2)n}$ and $|A| \ge 2^{(1-\delta)n}$, we get that $|B_\sigma| \le 2^{(h(\sigma/2) + o_\delta(1))n}$.
  Plugging this into \eqref{eqn:collision bound Bsigma} we see that
  \[
  \|b\|_2^2 \le \max_{\sigma} 2^{(1 + h(\sigma/2) - \sigma + o_\epsilon(1))n}\,.
  \]
  The expression $h(\sigma/2)-\sigma$ is maximized at $\sigma=2/5$, and we obtain
  \[
  \|b\|_2^2 \,\le\, 2^{(h(1/5)+3/5 + o_\epsilon(1))} \,\le\, 2^{(1.32195 + o_\delta(1))n}\,.
  \]
  Thus if $\delta$ is sufficiently small, we have $\|b\|_2^2 \le 2^{1.322n}$, as desired.
\end{proof}

\section{Proof of Theorem~\ref{thm:densityreduction}}\label{sec:density}

\begin{proof}[Proof of Theorem~\ref{thm:densityreduction}]

  Given oracle access to an algorithm that solves \subsetsum{}
  instance of density at least $1.003$ in constant time, we solve an
  arbitrary instance $w_1,w_2,\ldots,w_n,t$ of \subsetsum{} in time
  $O^*(2^{0.49991n})$ as follows.

  As Step 1, run the algorithm of Theorem~\ref{cor:manysums} for
  $\Theta^*(2^{0.49991n})$ timesteps. If it terminates within this
  number of steps, return YES if it found a solution and NO
  otherwise.
  Otherwise, as Step 2, run the preprocessing of Lemma~\ref{lemma:reduce
    t} with $B = 10 \cdot 2^{0.997n}$.  This yields a new instance
  with density $1/0.997 > 1.003$, which we solve using the presumed oracle for
  such instances.  If the oracle returns a solution, we verify that it
  is indeed a solution to our original instance and if so return YES.
  Otherwise we return NO.

  If there is no solution this algorithm clearly returns NO. If there
  is a solution and $|w(2^{[n]})|\geq 2^{0.997n}$, we find a solution
  with inversely polynomial probability in Step 1. If there is a
  solution and $|w(2^{[n]})|\leq 2^{0.997n}$, Property~\ref{hashlemma:correctness} of
  Lemma~\ref{lemma:reduce t} guarantees that the solution to the reduced
  instance is a solution to the original instance with probability
  $\Omega^*(1)$, and the oracle will then provide us with the solution.
\end{proof}

\section{Further Discussion}\label{sec:discussion}

Our original ambition was to resolve Open Question~\ref{q:fast} affirmatively by a combination of two algorithms that exploit small and large concentration of the sums, respectively. Since we only made some partial progress on this, it remains an intruiging question whether this approach can fulfill this ambition. In this section we speculate about some further directions to explore.

\paragraph{Exploiting Large Density}
For exploiting a density $1.003\leq d \leq 2$, the meet-in-the-middle technique~\cite{HorowitzSahni74} does not seem directly extendable. A different, potentially more applicable $O^*(2^{n/2})$ algorithm works as follows: pick a prime $p$ of order $2^{n/2}$, build the dynamic programming table that counts the number of subsets with sum congruence to $t$ mod $p$, and use this as a data structure to uniformly sample solutions mod $p$ with linear delay; try $O^*(2^{n/2})$ samples and declare a no-instance if no true solution is found (see also Footnote~\ref{footnote:manysoln}). As such, this does not exploit large density at all, but to this end one could seek a similar sampler that is more biased to smaller bins.

\paragraph{Sharper Analysis of Algorithm~\ref{alg:manysums}}
The analysis of Algorithm~\ref{alg:manysums} in Lemma~\ref{lem:binalgo}, and in particular the typical bottleneck $\beta(w)2^{(1.5-\gamma)\mu n}$ in the running time, is quite naive. For example, since we can pick $M$ as we like (and assume it generates many sums), for the algorithm to fail we need an instance where big bins are encountered by the algorithm for many choices of $M$. It might be a good approach to first try to extend the set of instances that can be solved `truly faster' in this way, e.g. to the set of all instances with $\beta(w)\leq 2^{(.5+\delta)n}$ for some small $\delta > 0$.

As an illustration of the looseness, let us mention that in a previous version of this manuscript, we used a more sophisticated analysis to show the following: there exists some $\delta >0$, such that if $|w(2^{[n]})| \geq 2^{(1-\delta)n}$, then
\[
	|\{(P,Q) \in \binom{[n]}{n/2}^2: w(P)+w(Q)= t \}| \leq 2^{0.5254n}.
\]
We used this to show that all instances with $|w(2^{[n]})| \geq 2^{(1-\delta)n}$ can be solved via a mild variant of Algorithm~\ref{alg:manysums} with $M=[n]$, indicating that Algorithm~\ref{alg:manysums} gives non-trivial algorithms even for large $M$.

\paragraph{Sharper Combinatorial Bounds}

Lemma~\ref{lem:sumsvsbin} and Lemma~\ref{lemma:large bin ub} seem to be rather crude estimates. In fact, we don't even know the following (again, borrowing notation from the proof of Proposition~\ref{prop:large bin l2}):

\begin{openquestion}\label{q:sumsvsbins}
  Suppose $|w(2^{[n]})| \geq 2^{(1-\epsilon)n}$. Can $\beta(w)$ and $\|b_{[n]}\|_2$ be bounded by $2^{o_\epsilon(1) n}$ and $2^{(0.5+o_\epsilon(1)) n}$, respectively? 
\end{openquestion}

Note that the second bound would follow from the first bound.
Furthermore, if the second bound holds, we would be able to solve, for
all $\epsilon >0$, all instances with $|\beta(w)| \geq
2^{(0.5+\epsilon)n}$ in time $\Os(2^{(0.5-\epsilon')n})$ for some
$\epsilon'>0$ depending on $\epsilon$, via the proof of
Theorem~\ref{thm:large bin easy}.

In recent work~\cite{AKKNUDCP} we proved the following modest progress
\begin{lemma}
  There exists $\delta >0$ such that if $A,B \subseteq \{0,1\}^n$ is a UDCP and $|A|\geq 2^{(1-\delta)n}$, then $|B|\leq 2^{0.4115n}$.
\end{lemma}

Plugging this into the proof of Lemma~\ref{lem:sumsvsbin}, this gives that $\beta(w) \leq 2^{(0.4115+o_\epsilon(1))n}$ in the setting of Open Question~\ref{q:sumsvsbins}. We would like to remark that improving this beyond $2^{(0.25+o_\epsilon(1))n}$ via Lemma~\ref{lem:sumsvsbin} is not possible since UDCP pairs $(A,B)$ with $|A|\geq 2^{(1-o(1))n}$ and $|B|\geq 2^{n/4}$ do exist~\cite{Kasami}.

One may also wonder whether we can deal with instances with $|w(2^{[n]})|\geq 2^{(0.5+\epsilon)n}$, for all $\epsilon >0$ by arguing $\beta(w)$ must be small but this does not work directly: there are instances with $|w(2^{[n]})|=3^{n/2}$ and $\beta(w)=2^{n/2}$ (the instance $1,1,3,3,9,9,27,27,\ldots$ has this, though is it easily attacked via Lemma~\ref{lemma:exploiting few sums}).

\subsection*{Acknowledgements}

This research was funded in part by
the Swedish Research Council, Grant 621-2012-4546 (P.A.), 
the European Research Council, Starting Grant 338077 
``Theory and Practice of Advanced Search and Enumeration'' (P.K.),
the Academy of Finland, Grant 276864 ``Supple Exponential Algorithms'' (M.K.), 
and by the NWO VENI project 639.021.438 (J.N.).
\small
\bibliographystyle{plain}
\bibliography{paper}

\end{document}